\documentclass[conference,a4paper]{IEEEtran}

\usepackage[utf8]{inputenc}
\usepackage[T1]{fontenc}
\usepackage{url}
\usepackage{ifthen}
\usepackage{cite}
\usepackage[cmex10]{amsmath} % Use the [cmex10] option to ensure complicance
\usepackage{stackengine}
\usepackage{gensymb}
\usepackage{graphics}
\usepackage{mypackage}
\usepackage{mathrsfs}

\usepackage{mathtools}
\usepackage{tabu}
\usepackage{float}
\usepackage{physics}
\usepackage{siunitx}
\usepackage{multicol}
\usepackage{amssymb}
\usepackage[nomain,acronym,shortcuts]{glossaries}

\usepackage{etoolbox}

\newcommand{\sir}{\mathrm{SIR}}
\newcommand{\Pb}{\mathbb{P}}
\newcommand{\Eb}{\mathbb{E}}

\newcommand{\Lc}{\mathcal{L}}

\DeclarePairedDelimiter\floor{\lfloor}{\rfloor}

 \makeglossaries
\newcommand*{\acro}[3][]{\newacronym[#1]{#2}{#2}{#3}}

\newtheorem{corollary}{Corollary}
\newtheorem{theorem}{Theorem}
\theoremstyle{approximation}

\newcommand{\black}{\textcolor{black}}

\DeclareMathOperator*{\R}{\mathbb{R}}

\let\mybibitem\bibitem
\renewcommand{\bibitem}[1]{%
  \ifstrequal{#1}{nature}
    {\color{blue}\mybibitem{#1}}
    {\color{black}\mybibitem{#1}}%
}

\acro{OFDM}{orthogonal frequency-division multiplexing}
\acro{OFDMA}{orthogonal frequency-division multiple access}
\acro{MNO}{mobile network operator}
\acro{RA}{resource allocation}
\acro{SC-FDMA}{single carrier frequency division multiple access}
\acro{CR}{cognitive radio}
\acro{RFIC}{radio frequency integrated circuit} 
\acro{SDR}{software defined radio}
\acro{SDN}{software defined networking}
\acro{su}{secondary user}
\acro{QoS}{quality-of-service}
\acro{USRP}{universal software radio peripheral}
\acro{GSM}{Global System for Mobile Communications}
\acro{TDMA}{time-division multiple access}
\acro{FDMA}{frequency-division multiple access}
\acro{GPRS}{General Packet Radio Service}
\acro{MSC}{Mobile Switching Centre}
\acro{BSC}{Base Station Controller}
\acro{UMTS}{Universal Mobile Telecommunications System}
\acro{WCDMA}{wide-band code division multiple access}
\acro{CDMA}{code division multiple access}
\acro{LTE}{Long Term Evolution}
\acro{PAPR}{peak-to-average power rating}
\acro{HetNet}{heterogeneous networks}
\acro{PHY}{physical layer}
\acro{MAC}{medium access control}
\acro{AMC}{adaptive modulation and coding}
\acro{MIMO}{multiple-input multiple-output}
\acro{M-MIMO}{massive MIMO}
\acro{RAT}{radio access technology}
\acro{VNI}{visual networking index}
\acro{RB}{resource block}
\acro{UE}{user equipment}
\acro{CQI}{channel quality indicator}
\acro{HD}{half duplex}
\acro{IBFD}{in-band full duplex}
\acro{SIC}{self-interference cancellation}
\acro{SI}{self-interference}
\acro{BS}{base station}
\acro{FBMC}{filter bank multi-carrier}
\acro{UFMC}{universal filtered multi-carrier}
\acro{SCM}{single carrier modulation}
\acro{isi}{inter-symbol interference}
\acro{FTN}{faster-than-nyquist}
\acro{M2M}{machine-to-machine}
\acro{MTC}{machine type communication}
\acro{mmWave}{millimeter wave}
\acro{BF}{beamforming}
\acro{LoS}{line-of-sight}
\acro{NLoS}{non-line-of-sight}
\acro{CAPEX}{capital expenditure}
\acro{OPEX}{operational expenditure}
\acro{ICT}{information and communications technology}
\acro{SP}{service providers}
\acro{InP}{infrastructure providers}
\acro{MVNP}{mobile virtual network provider}
\acro{MVNO}{mobile virtual network operator}
\acro{NFV}{network function virtualization}
\acro{VNF}{virtual network functions}
\acro{C-RAN}{cloud radio access network}
\acro{RAN}{radio access network}
\acro{BBU}{baseband unit}
\acro{RRH}{remote radio head}
\acro{TDD}{time-division duplexing}
\acro{FDD}{frequency-division duplexing}
\acro{GFDM}{generalized frequency division multiplexing}
\acro{CSI}{channel state information}
\acro{FFT}{fast Fourier transform}
\acro{IFFT}{inverse FFT}
\acro{CFO}{carrier frequency offset}
\acro{CoMP}{coordinated multipoint}
\acro{D2D}{device-to-device}
\acro{OOB}{out-of-band}
\acro{TTI}{transmission time interval}
\acro{DUE}{D2D user equipment}
\acro{DAS}{distributed antenna system}
\acro{ICIC}{inter-cell interference coordination}
\acro{PDF}{probability distribution function}
\acro{KPI}{key performance indicator}
\acro{SBS}{small base station}
\acro{MBS}{macro base station}
\acro{SCN}{small cell network}
\acro{FIFO}{first in first out}
\acro{VCC}{virtual cache center}
\acro{UAV}{unmanned aerial vehicles}
\acro{MPSQ}{multiclass processor sharing queue} 
\acro{EE}{energy efficiency}
\acro{SIR}{signal-to-interference ratio}
\acro{SINR}{signal-to-noise-plus-interference ratio}
\acro{PPP}{Poisson point process}
\acro{PCP}{Poisson cluster process}
\acro{HCP}{hard-core placement} 
\acro{TCP}{Thomas cluster process }
\acro{CPF}{caching popular files}
\acro{GCA}{greedy caching algorithm }
\acro{RC}{random caching }
\acro{PC}{probabilistic caching}
\acro{5G}{fifth generation}
\acro{MEC}{mobile edge computing}
\acro{AP}{access point}
\acro{VoD}{video-on-demand}
\acro{EPC}{evolved packet core}
\acro{QoE}{quality-of-experience}
\acro{CDN}{content delivery networks}
\acro{F-RAN}{fog-radio access network}
\acro{AR}{augmented reality}
\acro{VR}{virtual reality}
\acro{4C}{computing, caching, communication, and control}
\acro{ABR}{Adaptive BitRate}
\acro{ILP}{Integer Linear Program}
\acro{MILP}{Mixed Integer Linear Program}
\acro{MINLP}{Mixed Integer Non-Linear Program}
\acro{BC}{broadcast channel}
\acro{MDS}{maximum distance separable}
\acro{RS}{Reed-Solomon}
\acro{NC}{network coding}
\acro{MSR}{minimum storage regenerating}
\acro{Coop-MIMO}{cooperative \ac{mimo}}
\acro{UDN}{ultra-dense networks}
\acro{ETSI}{European Telecommunications Standards Institute}
\acro{MCP}{Matern cluster process}
\acro{MD-CoMP}{macrodiversity CoMP transmission}
\acro{JT-CoMP}{joint transmission CoMP}
\acro{CoMP-JT}{coordinated multipoint joint transmission}
\acro{MDSD}{multiple devices to the single device}
\acro{PMF}{probability mass function}
\acro{RV}{random variable}
\acro{i.i.d.}{independently and identically distributed}
\acro{MBMS}{multimedia broadcasting multicasting service}
\acro{CCDF}{complementary cumulative distribution function}
\acro{CDF}{cumulative distribution function}
\acro{PGFL}{probability generating functional}
\acro{KKT}{Karush-Kuhn-Tucker}
\acro{PGF}{point generating function}
\acro{BPP}{binomial point process}
\acro{3GPP}{3rd Generation Partnership Project}
\acro{FAA}{Federal Aviation Administration}
\acro{UAS}{unmanned aircraft systems}
\acro{IoT}{Internet-of-things}
\acro{CNPC}{control and non-payload communication}
\acro{FHD}{full high definition}
\acro{MGF}{moment generating function}
\acro{3D}{three-dimensional}
\acro{2D}{two-dimensional}
\acro{1D}{one-dimensional}
\acro{CB}{conjugate beamforming}
\acro{AU}{aerial user}
\acro{GU}{ground user}
\acro{CLT}{central limit theorem}
\acro{SE}{spectral efficiency}
\acro{MRT}{maximum ratio transmission}
\acro{ZFBF}{zero-forcing beamforming}
\acro{GBS}{ground base station}
\acro{C&C}{command and control}
\acro{ASE}{area spectral efficiency}
\acro{RWP}{random waypoint}
\acro{AGL}{above ground level}
\acro{w.r.t.}{with respect to}
\acro{GTA}{ground-to-air}
\acro{AUE}{aerial user equipment}
\acro{GUE}{ground user equipment}
\acro{UB}{upper bound}
\acro{LB}{lower bound}
\acro{ABS}{aerial base station}
\acro{UAV-UE}{UAV-\ac{UE}}
\acro{SCDP}{successful content delivery probability}
\acro{HARP}{highest average received power}
\acro{ULA}{uniform linear array}
\acro{GPP}{Gaussian Poisson process}
\acro{NSD}{nearest serving device}
\acro{NCP}{nearest content provider}
\acro{RSD}{randomly serving device}
\acro{RSCP}{randomly-selected content provider}
\acro{BCD}{block coordinate descent}
\acro{NOMA}{non-orthogonal multiple access}
\acro{MPC}{most popular content}
\acro{LPC}{least popular content}
\acro{AI}{artificial intelligence}
\acro{ML}{machine learning}
\acro{TPP}{temporal point process}
\acro{PP}{point process}
\acro{DPP}{Determinantal point process}
\acro{DPPL}{DPP-based learning}
\acro{DRL}{deep reinforcement learning}
\acro{RL}{reinforcement learning}
\acro{RNN}{recurrent neural network}
\acro{RMTPP}{recurrent marked TPP}
\acro{NR}{New Radio}
\acro{HO}{handover}
\acro{RLF}{radio link failure}
\acro{RSRP}{reference signal received power}
\begin{document}
\title{Performance Analysis of Mobile Cellular-Connected Drones under Practical Antenna Configurations}
% Sky Mobility-Aware Cellular Networks: CoMP meets Cellular-Connected Drones for Enhanced Mobility Support
% UAVs
 %\author{Author 1, Author 2, and Author 3}
% \author{Ramy Amer, Walid Saad, and Nicola Marchetti}

%\author{\IEEEauthorblockN{Ramy Amer\IEEEauthorrefmark{1},			
%Walid Saad\IEEEauthorrefmark{2},
%Nicola Marchetti\IEEEauthorrefmark{1}}\\
%\IEEEauthorblockA{\IEEEauthorrefmark{1}CONNECT, Trinity College, University of Dublin, Ireland}\\
%\IEEEauthorblockA{\IEEEauthorrefmark{2}Wireless@VT, Bradley Department of Electrical and Computer Engineering, Virginia Tech, Blacksburg, VA, USA}\\
%\IEEEauthorblockA{email:\{ramyr, nicola.marchetti\}@tcd.ie, walids@vt.edu}}

 \author{\IEEEauthorblockN{Ramy Amer\IEEEauthorrefmark{1},			
Walid Saad\IEEEauthorrefmark{2},
Boris Galkin\IEEEauthorrefmark{1}, 
Nicola Marchetti\IEEEauthorrefmark{1}}
%\\
\IEEEauthorblockA{\IEEEauthorrefmark{1}CONNECT, Trinity College, University of Dublin, Ireland,}
%\\
\IEEEauthorblockA{\IEEEauthorrefmark{2}Wireless@VT, Bradley Department of Electrical and Computer Engineering, Virginia Tech, Blacksburg, USA,}
\IEEEauthorblockA{Emails:\{ramyr, galkinb, nicola.marchetti\}@tcd.ie, walids@vt.edu.}
\thanks{This publication has emanated from research conducted with the financial support of Science Foundation Ireland (SFI) and is co-funded under the European Regional Development Fund under Grant Number 13/RC/2077, and the U.S. National Science Foundation under Grants CNS-1836802 and IIS-1633363.} }
\maketitle
\begin{abstract}
Providing seamless connectivity to unmanned aerial vehicle user equipments (UAV-UEs) is very challenging due to the encountered line-of-sight interference and reduced gains of down-tilted base station (BS) antennas. For instance, as the altitude of UAV-UEs increases, their cell association and handover procedure become driven by the side-lobes of the BS antennas. In this paper, the performance of cellular-connected UAV-UEs is studied under 3D practical antenna configurations. Two scenarios are studied: scenarios with static, hovering UAV-UEs and scenarios with mobile UAV-UEs. For both scenarios, the UAV-UE coverage probability is characterized as a function of the system parameters. The effects of the number of antenna elements on the UAV-UE coverage probability and handover rate of mobile UAV-UEs are then investigated. Results reveal that the UAV-UE coverage probability under a practical antenna pattern 
is worse than that under a simple antenna model. Moreover, vertically-mobile UAV-UEs are susceptible to \emph{altitude handover} due to consecutive crossings of the nulls and peaks of the antenna side-lobes.

%%%%%%%%%%%%%%%%%%%%%%%%%%%%%%%%%%%%%%%%%%%%%%%%%%%%%%7510820%%%%%%%%%%%%%%%
\end{abstract}
\begin{IEEEkeywords}
UAV-UEs, 3D antenna, altitude handover.	%  stochastic geometry		management
\end{IEEEkeywords}

\vspace{-0.2 cm}
\section{Introduction}	
\vspace{-0.1 cm}
A tremendous increase in the use of \acp{UAV}, i.e., drones, in a wide range of applications, ranging from aerial surveillance and safety to product delivery, is anticipated  in the foreseeable future \cite{8660516,mozaffari2016unmanned,saad2019vision,kishk20193}. In such applications, UAVs need to communicate with each other as well as with ground \acp{UE} using wireless cellular connectivity.  

Cellular-connected UAVs have attracted attention in cellular network research in both industry and academia due to their ability to ubiquitously communicate \cite{azari2017coexistence,lin2018sky,85317111,amer2020caching,8756296}.  % amer2019mobility,8756296,amer2019mobility,8528463
However, the \ac{BS} antennas of current cellular networks are tilted downwards to provide connectivity to  ground \acp{UE} rather flying UAV-UEs \cite{azari2017coexistence}. Hence, UAV-UEs have to be served from the side-lobes of the BS antennas. Moreover, the UAV-UEs, especially at high altitudes, are dominated by \ac{LoS} communication links. These key characteristics of the UAV-UE communications pose new technical challenges on their cell association, handover procedure, and the overall achievable performance \cite{amer2019mobility}. 
		% banagar20193gpp

% Moreover, the UAV-UEs, especially at high altitudes, are dominated by \ac{LoS} interference and reduced \acp{BS} antenna gain \cite{azari2017coexistence}. These key characteristics of UAV-UE communication channels pose new technical challenges on \blue{their association \cite{8422685} \ac{LoS} and gain of the antenna.} 

%For instance, cellular-connected UAV communication possesses substantially different characteristics that pose new technical challenges which include: dominance of \ac{LoS} interference and reduced \acp{BS} antenna gain \cite{azari2017coexistence}.
%
%% To enable ubiquitous communications to the UAVs, cellular-connected UAV, whereby UAVs are integrated into the cellular network as new UAV-UEs, is a promising technology that has drawn significant attention \cite{azari2017coexistence,amer2019mobility,85317111,8528463}.

In this regard, the authors in \cite{8692749} studied the feasibility of supporting drone operations using existing cellular infrastructure. It is shown that, under a simple antenna model, the cell association heavily depends on the availability of \ac{LoS} links to the serving BS. In \cite{8713514}, similar results are verified for air-to-ground communication between UAV-BSs and ground \acp{UE}. Moreover, in \cite{amer2019mobility}, we showed that coordinated transmissions can effectively mitigate the effects of \ac{LoS} interference of high-altitude UAV-UEs. While the works in \cite{8692749,8713514,amer2019mobility} considered the possibility of \ac{LoS} communication, they assumed a simple antenna configuration that is modeled as a step function of two gain values, namely, main- and side-lobe gains. However, practical antenna patterns resemble a sequence of main- and side-lobes with nulls between consecutive lobes. Such a \ac{3D} antenna pattern plays a crucial rule on the UAV-UE cell association and handover procedure, which is ignored in most of the prior works \cite{8692749,8713514,amer2019mobility}. Performance analysis of UAV-UEs under \ac{3D} practical antenna models was done in the recent works  \cite{8528463,8756719,abs-1804-04523,HCC:3325421.3329770}. For instance, based on system-level simulations, the authors in \cite{8528463} showed that the cell association of UAV-UEs is mainly dependent on the side-lobes. Moreover, in \cite{8756719}, the authors characterized the association probability and \ac{SIR} under nearest-distance and maximum-power based associations. 

The performance of mobile UAV-UEs  under practical antenna patterns has been studied in recent works  \cite{abs-1804-04523} and \cite{HCC:3325421.3329770}. For instance, based on system-level simulation, the authors in \cite{abs-1804-04523} showed that, due to the \ac{LoS} propagation conditions to many interfering cells, it is difficult for the UAV-UEs to establish and maintain connections to the network, which also leads to increased handover failure rates. Moreover, based on experimental trials in \cite{HCC:3325421.3329770}, the authors showed that drones are subject to frequent handovers once the typical flying altitude is reached. However, the results presented in these works are based on simulations and measurements. Also, although the work in \cite{8756719} has characterized the \ac{SINR} at UAV-UEs served from \ac{3D} antennas, there was no characterization of important performance metrics such as the coverage probability. Moreover, this work only considered a static UAV-UE scenario. \emph{As a first step in this direction, we seek to characterize the performance of static and mobile UAV-UEs under practical antenna configurations}.

Compared with this prior art \cite{8692749,8713514,amer2019mobility,8528463,8756719,abs-1804-04523,HCC:3325421.3329770}, the main contribution of this paper is a rigorous analysis that provides an in-depth understanding of the performance of UAV-UEs under practical antenna configurations. In particular, we consider a network of ground BSs equipped with more than two antenna array elements to provide cellular connectivity to UAV-UEs. For this network, we characterize the coverage probability for two scenarios, specifically, static and mobile UAV-UEs. Moreover, we investigate the handover rate of mobile UAV-UEs in order to provide important design guidelines and understand novel handover aspects of UAV-UEs, such as the altitude handover. Our results show that the number of antenna elements  controls the UAV-UE handover rate, while its impact on the coverage probability is marginal if the handover cost is low.  
% especially at low
% \emph{To the best of our knowledge, this paper provides the first rigorous analysis for static and mobile UAV-UEs under practical antenna configuration, and provides a novel concept of altitude handover.}
% 
% We investigate the achievable performance of UAV-UEs under a BS antenna configuration suggested by 3GPP for two scenarios, namely, static and vertically mobile UAV-UEs. For each scenario, we characterize the coverage probability of UAV-UEs under nearest  and maximum received power association schemes.
 \begin{figure} [!t]	%[!t] %%    [htbp]		uav-comp_antenna
  \vspace{-0.3 cm}
\centering
\includegraphics[width=0.35 \textwidth]{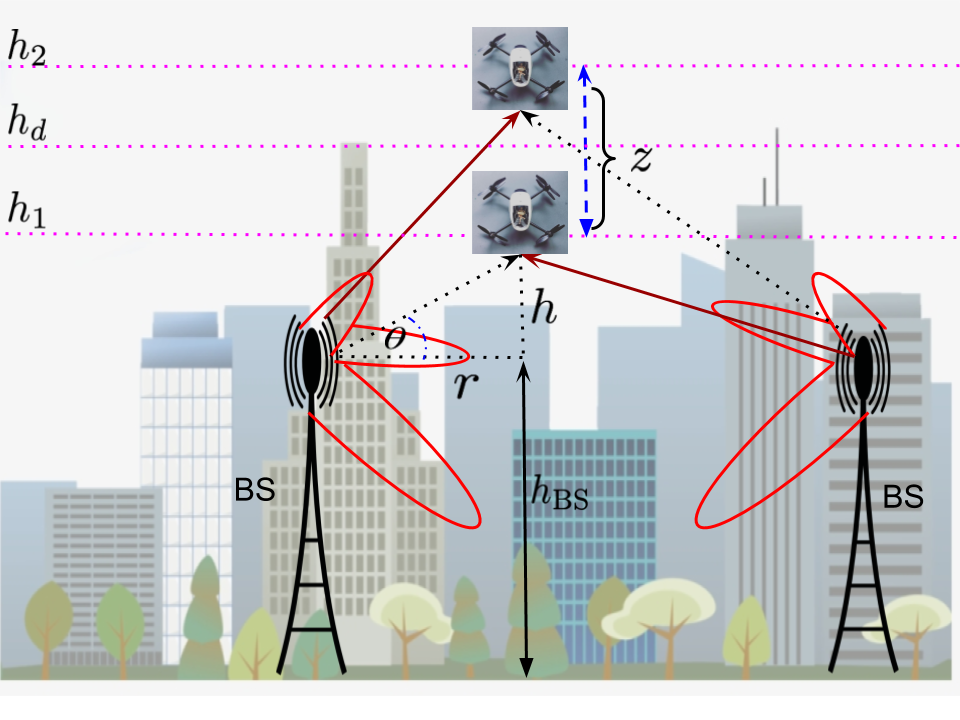}	%letter	vert_mob.png	
\caption {Illustration of the proposed system model in which 3D antenna-equipped ground BSs serve high-altitude static (or mobile) UAV-UEs.}
\label{system-model-comp0}
 \vspace{-0.4 cm}
\end{figure}
% \red{The elevation angle $\theta={\rm arctan}(\frac{h}{r})$.}
% Concept of Altitude handover for drones. Due to the effect of antenna side-lobes and nulls between these lobes, an UAV-UE at altitude $h_1$ is associated to the right BS while it is associated to the left BS at altitude $h_2$.

 \vspace{-0.1 cm}
\section{System Model}
 \vspace{-0.2 cm}
\subsection{Network Model}
We consider a downlink transmission scenario from a terrestrial cellular network to cellular-connected UAV-UEs. We assume that ground \acp{BS} are distributed according to a \ac{2D} homogeneous \ac{PPP}  $\Phi_b=\{ b_i \in \mathbb{R}^2, \forall i \in \mathbb{N}\}$ with intensity $\lambda_b$. All \acp{BS} have the same transmit power $P_t$ and are deployed at the same height $h_{\rm BS}$. We consider a number of UAV-UEs that can be either  static or mobile based on the application. Static UAV-UEs hover at a fixed altitude $h_d$, while mobile ones can make up and down movements within two altitude thresholds $h_1$, and $h_2$. We set $h_d=\frac{h_1+h_2}{2}$, i.e., $h_d$ is the mean flying altitude of mobile UAV-UEs. As shown in Fig.~\ref{system-model-comp0}, we consider high-altitude UAV-UEs where $h_1$, $h_d$, and $h_2$ are above the BS height $h_{\rm BS}$. Each UAV-UE has a single antenna and receives downlink signals from a ground \ac{BS}. Each ground BS is equipped with a directional antenna array composed of $N_t$ vertically-placed elements. Two association schemes are considered for the UAV-UEs: Nearest and \ac{HARP} associations.

% Next, we first characterize the serving distance distribution, and then, we employ it to derive upper and lower bounds on the coverage probability of static UAV-UEs.
% 
%We use the \ac{LoS} ball to model the blockage effect as recently adopted in the literature [9], [20]. Specifically, we define a \ac{LoS} radius $R$, which represents the distance between a receiver and its nearby blockages, and the \ac{LoS} probability of a certain link is one within $R$ and zero outside the radius.  

%\subsection{\red{\textbf{Angles and gain of main-lobes and side-lobes}}}
\vspace{-0.0 cm}
\subsection{Channel Model}
\vspace{-0.1 cm}
We consider a wireless channel that is characterized by both small-scale and large-scale fading. For the large-scale fading, the channel between a ground \ac{BS} and an UAV-UE is described by the \ac{LoS} and \ac{NLoS} components, which are considered separately along with their probabilities of occurrence \cite{azari2017coexistence}. This assumption is apropos for such ground-to-air channels that  often exhibit \ac{LoS} communication   \cite{azari2017coexistence} and \cite{8713514}. For small-scale fading, we adopt a Nakagami-$m_v$ model as done in \cite{azari2017coexistence} for the channel gain, whose \ac{PDF} is given by:
%\begin{align}
$f(\omega) = \frac{2m_v^{m_v} \omega^{2m_v-1}}{\Gamma(m_v)} e^{-m_v \omega^2}$.    
%$\end{align}
The fading parameter $m_v$ is assumed to be an integer for tractability, where $v\in \{l,n\}$ accounts for \ac{LoS} and \ac{NLoS} communications, respectively. Given that $\omega \sim $ Nakagami$(m_v)$, it directly follows that the channel gain power  $\chi=\omega^2 \sim \Gamma(m_v,1/m_v)$, which represents a Gamma \ac{RV} whose shape and scale parameters are $m_v$ and $\frac{1}{m_v}$, respectively. Hence, the \ac{PDF} of channel power gain distribution will be:  
%\begin{align}
$f(\chi) = \frac{m_v^{m_v} \chi^{m-1}}{\Gamma(m_v)} {\rm exp}\big(-m_v\chi\big)$.  
%\end{align}

%\blue{The incorporation of the blockages induces different path loss laws for \ac{LoS} and \ac{NLoS} links.} 
%%
%We will focus on the analysis where the \red{typical static and mobile receiver} is associated with a ground BS that is assumed to always be of \ac{LoS} link, and the interference stems from both \ac{LoS} and \ac{NLoS} interferers. The relevant transmitters thus form a PPP, denoted as $\Phi_b$, with density $\lambda_b$ in a disk of radius $R$ centered at the origin.  

\ac{3D} blockage is characterized by the mean number of $\eta$ buildings/\SI{}{km}$^{2}$, the proportion $a$ of the total land area occupied by buildings,  and the height of buildings modeled by a Rayleigh \ac{PDF} with a scale parameter $c$. The probability of having a \ac{LoS} communication from a \ac{BS} at horizontal-distance $r$ from an UAV-UE  is hence given, similar to  \cite{8692749} and \cite{8713514}, as %\red{Jacek delta - Tabassum Series expansions}
\begin{align}
\label{prob-los}
\Pb_{l}(r) = \prod_{n=0}^{{\rm max}(0,o-1)}\Bigg[1 - {\rm exp}\Big(- \frac{\big(h_{\rm BS} + \frac{h(n+0.5)}{o}\big)^2}{2c^2}\Big) \Bigg], 
\end{align}		% ${\rm max}(p-1,0)$				% IT WAS o+11111
where $h$ is the difference between the UAV-UE altitude and \ac{BS} height and $o=\floor{r\sqrt{a\eta}}$. 
% Different terrain structures and environments can be considered by varying the tuple $(a,\eta,c)$. 
%, which depends on whether the UAV-UE is static or mobile
% the power radiation pattern for a conventional beam-former is the 

\vspace{-0.2 cm}
\subsection{Antenna Model}
The ground BSs are equipped with directional antennas of fixed radiation patterns and with a down-tilt angle $\varphi$. This is typically achieved by equipping the BS with a \ac{ULA} of $N_t$ vertically-placed elements, that are assumed to be omni-directional along the horizontal dimension \cite{3gpp-2017}. Along the vertical dimension, the power radiation pattern is equal to the array factor times the radiation pattern of a single antenna. The $N_t$ antenna elements are equally spaced with adjacent elements separated by half of the wavelength. With the down-tilt angle $\varphi$, the overall array gain in the direction $\theta$ is given by \cite{8756719}:  
\begin{align}
% \label{ant-gain}
G(\theta) = \underbrace{\frac{1}{N_t} \frac{{\rm sin}^2\frac{N_t\pi}{2}\big({\rm sin}(\theta)-{\rm sin}(\varphi)\big)}{{\rm sin}^2\frac{\pi}{2}\big({\rm sin}(\theta)-{\rm sin}(\varphi)\big)}}_{A_f(\theta)} \times 
\underbrace{ 10^{{\rm min}(-1.2(\frac{\theta}{\delta})^2,\frac{G_m}{10})} }_{G_e(\theta)},
\nonumber 
\end{align}
where $G_m$ gives the threshold for antenna nulls, $A_f(\theta)$ is the array factor of the \ac{ULA}, $G_e(\theta)$ is the element power gain of each antenna along the vertical dimension, i.e., the elevation angle $\theta$, and $\delta$ is the half power beamwidth. For simplicity, we assume that the gain $G_e(\theta)$ is a positive constant on the range of the elevation angle $\theta\in[0,\frac{\pi}{2}]$, and 0 otherwise (i.e., zero front-to-back power ratio as in \cite{vu2014analytical}).\footnote{Note that the elevation angle is bounded as $\theta\in[0,\frac{\pi}{2}]$ from the assumption of high-altitude UAV-UEs, i.e., $h_1,h_d,h_2>h_{\rm BS}$.} 
% the BS always steers to the direction of the served UAV-UE and the gain
% \red{array factor} of the \ac{ULA}
% each antenna element has the element power gain in dB given by
% To resemble the antennas of most existing cellular BSs as specified by 3GPP [2], 
% conventional beamforming technique
%
%
%
%Along the vertical dimension, each antenna element has the element power gain in $\SI{}{dB}$ given by
%\begin{align}
%G_e(\theta) = -{\rm min}\Bigg(12\big(\frac{\theta}{\Omega}\big)^2,G_m\Bigg),
%\end{align}
%where $\theta\in[-90^{\circ},90^{\circ}]$ is the elevation angle in degree, $\Omega$ is the half power beamwidth  and $G_m$ gives the threshold for antenna nulls. 
%
%
%Let $G_f(\theta)=10{\rm log}10(A_f(\theta))^2$ denote the array power gain in $\SI{}{dB}$. The overall antenna gain at elevation angle $\theta$ is given:
%\begin{align}
%G_v(\theta)= G_e(\theta) + G_f(\theta).
%\end{align}
%To calculate the gain of the antenna as a ratio rather than in $\SI{}{dB}$, we have $G_e(\theta)$ as a ratio given by (neglecting $G_m$)
%\begin{align}
%G_e(\theta)=10^{-1.2(\frac{\theta}{\Omega})},
%\end{align}
%
%And the total gain as a ratio will be: 
%\begin{align}
%\label{antenna-gain}
%G_v(\theta)= \frac{1}{N} \frac{{\rm sin}^2\frac{N\pi}{2}\big({\rm sin}(\theta)\big)}{{\rm sin}^2\frac{\pi}{2}\big({\rm sin}(\theta)\big)} \times 10^{-1.2(\frac{\theta}{\Omega})} 
%\end{align}			
% $h=h_d-h_{\rm BS}$ and 
From Fig.~\ref{system-model-comp0}, $\theta={\rm arctan}(\frac{h}{r})$, where $r$ is a realization of the \ac{RV} $R$ which represents the horizontal distance between a ground \ac{BS} and the projection of an UAV-UE. If $G_e(\theta)$ is set to one, and for a zero down-tilt angle, the antenna array gain is simplified to: 			% of the ULA
  \begin{align}
\label{antenna-gain} 
G(r,h) &= \frac{1}{N_t} \frac{{\rm sin}^2\frac{N_t\pi}{2}\big({\rm sin}({\rm arctan}(\frac{h}{r}))\big)}{{\rm sin}^2\frac{\pi}{2}\big({\rm sin}({\rm arctan}(\frac{h}{r}))\big)}. 
\end{align}

Hence, the antenna gain plus path loss for the \ac{LoS} and \ac{NLoS} components will be: $\zeta_v(r)= A_v G(r,h) \big(r^2 + h^2\big)^{-\alpha_v/2}$, where $v\in\{l,n\}$, $\alpha_{l}$ and $\alpha_{n}$ are the path loss exponents, and  $A_{l}$ and  $A_{n}$ are the path loss constants at $\sqrt{r^2+h^2} = \SI{1}{m}$ for the \ac{LoS} and \ac{NLoS}, respectively. 

Having defined our system model, next, we will study the performance of UAV-UEs for two scenarios: static and mobile UAV-UEs. For each scenario, we investigate the  UAV-UE coverage probability under two association schemes, namely, nearest and \ac{HARP} associations. The coverage probability is defined as the probability that the received \ac{SIR} is higher than a target threshold $\vartheta$.

 % Finally, since the \ac{LoS} probability decreases with the distance from (\ref{prob-los}), the \ac{LoS} assumption becomes less practical for far but interfering \acp{BS}.
%As discussed earlier, UAV-UEs exhibit a \ac{LoS} component which becomes dominant at relatively high altitudes. The \ac{LoS} probability in (\ref{prob-los}) represents a delta function that goes from one to zero as $r_i$ increases. This implies that the probability of \ac{LoS} communication from close \acp{BS} is higher than that of remote \acp{BS}. Hence, we consider that the desired signal is dominated by its \ac{LoS} component where $v=l$, $m_v=m_l$, and $P_v(r_i) =\sqrt{P_t}  \zeta_l(r_i)^{0.5}$. However, for the interfering signal, both \ac{LoS} and \ac{NLoS} components exist and, thus, we have: $P_v(u_j) =\sqrt{P_t}  \zeta_v(u_j)^{0.5}$, $v\in\{l,n\}$. This is due to the fact that, as the \ac{LoS} probability decreases with the interfering distance $u_j$, the \ac{LoS} assumption becomes less practical for far but interfering \acp{BS}. 

\vspace{-0.1 cm}
\section{Coverage Probability of Static UAV-UEs}
We assume static UAV-UEs hovering at a fixed altitude $h_d$, hence, we have $h=h_d-h_{\rm BS}$ in (\ref{prob-los}). Given that a \ac{PPP} is translation invariant with respect to the origin, we conduct the coverage analysis for a UAV-UE located at the origin in $\R^2$, called the typical UAV-UE  \cite{haenggi2012stochastic}. 
\vspace{-0.2 cm}
\subsection{Nearest Association}
\label{sec:near-assoc}
\vspace{-0.1 cm}

%\begin{assumption}
%\blue{\rm 
To simplify the analysis, we only consider probabilistic \ac{LoS}/\ac{NLoS} links for interfering BSs, while the serving BS has a dominant \ac{LoS} component. This is because, at high altitude, UAV-UEs will have an LoS-dominated channel toward nearby, serving BSs. However, UAV-UE interference at far-away BSs will not be LoS dominated. Moreover, the study of \ac{LoS}-based association is considered in prior works, e.g., \cite{8713514}, while we are particularly focused on the impact of practical antennas.  
% particularly focused on the antenna-based cell association
% and handover  
%\end{assumption}

We denote the horizontal distance from the typical UAV-UE to its ground \ac{BS} by $r_0$. By the \ac{PPP} assumption, $f_{R_0}(r_0)=2\pi r_0 \lambda e^{-\pi\lambda r_0^2}$. Conditioning on $R_0=r_0$, and neglecting the noise, the received $\sir$ at the typical UAV-UE will be:  
\begin{align}
\label{sir-pwr}
\Upsilon_{|r_0} &= \frac{\chi_0 \zeta_l(r_0)  }{I},  
\end{align} 		% , since all BSs transmit with the same power $P_t$
% The conditional coverage probability given a \ac{NLoS} serving ground BS is calculated in the same way.
where $I$ is the aggregate interference, $\chi_0$ is the Nakagami-$m_l$ fading power, and $\zeta_l(r_0)$ represents the antenna gain plus path loss from the serving BS. The serving and interfering signals in (\ref{sir-pwr}) are normalized to the transmit power $P_t$. The UAV-UE coverage probability is characterized in the next Theorem.
\begin{theorem}
\label{lem-stat-near}
The static UAV-UE coverage probability under nearest association is given by:  
\begin{align} 
\label{cov-prob-near12}
\mathbb{P}_{c} = \int_{r_0=0}^{\infty} \mathbb{P}_{c|r_0} f_{R_0}(r_0) \dd{r_0}, 
\end{align} 
where $\mathbb{P}_{c|r_0}$ is the UAV-UE conditional coverage probability: 	% given by
\begin{align}
\mathbb{P}_{c|r_0} & \overset{}{=}   \sum_{i=0}^{m_l-1} \frac{(-\varpi_l)^i}{i!}  \sum_{i_n+i_l=i}
\frac{i!}{i_n!i_l!}
 \frac{\varpi_l^{i_n}}{i_n!}
\frac{\partial^{i_n}}{\partial\varpi_l^{i_n}} \Lc_{I_n}(\varpi_l)  
\nonumber \\
\label{los-nlos-LT0}
& \quad \quad\quad \quad\quad \quad\quad \quad  \quad\quad \quad
\times \frac{\varpi_l^{i_l}}{i_l!}
\frac{\partial^{i_l}}{\partial\varpi_l^{i_l}} \Lc_{I_l}(\varpi_l)  ,
\end{align}
and $\varpi_l=\frac{\vartheta d_0^{\alpha_l}m_l}{A_lG(r_0,h)}$, $d_0=\sqrt{h^2+r_0^2}$; $\Lc_{I_l}$ and $\Lc_{I_n}$ are the Laplace transforms of the \ac{LoS} and \ac{NLoS} interference, respectively. The Laplace transform of \ac{LoS} interference is then given by  
$\Lc_{I_l}(\varpi_l)  = e^{- \gamma(\varpi_l)}$, where $\gamma(\varpi_l)=2\pi\lambda_b \times$ 
\begin{align}
\sum_{j=\floor{r_0\sqrt{a\eta}}}^{\infty} \Pb_{l}(\frac{j}{\sqrt{a\eta}})  
\int_{{\rm max}(r_0,\frac{j}{\sqrt{a\eta}})}^{\frac{j+1}{\sqrt{a\eta}}}  		% {\rm max}(r_0,
  \Big(1 - \big( \frac{m_l}{m_l + \varpi_l \zeta_l(r)} \big)^{m_l}  \Big) r \dd{r},  
 \nonumber  
  \end{align}
and the Laplace transform of \ac{NLoS} interference is calculated following in the same manner. 
\end{theorem}
\begin{proof}
The sketch of the proof is found in the Appendix. 
\end{proof}
% and $\Pb_{l}(j)$ is calculated from (\ref{prob-los}).  % 
% with $o=\max(0,\floor{j\sqrt{a\eta}-1})$    \frac{j}{1000} $o=\floor{j-1}$

Since it is hard to directly obtain insights from (\ref{cov-prob-near12}) on the effect of the practical antenna gain,  several results based on (\ref{cov-prob-near12}) will be shown in Section \ref{num-result} to provide key design guidelines. Moreover, we show that the coverage probability does not scale with $N_t$, which implies that increasing the number of antenna elements has a marginal effect on the coverage probability. 
% key practical insights		key practical
\begin{corollary}
\label{eff-of-Nt}
The coverage probability of static UAV-UEs does not scale asymptotically with the number of antenna elements.
%$N_t$
\end{corollary}
\begin{proof} To study scalability, we consider a simple case with $m_v=A_v=1$, $v\in\{l,n\}$, hence, $\varpi_l=\frac{\vartheta d_0^{\alpha_l}}{G(r_0,h)}$. We also assume that $N_t$ is small and $\frac{h}{r}\to 0$, which is a reasonable assumption for sparsely-deployed networks. The conditional coverage probability $ \mathbb{P}_{c|r_0}$ is then simplified to:
\begin{align} 
&\mathbb{P}_{c|r_0}=  e^{-2\pi\lambda_b \int_{r=r_0}^{\infty} 
  \frac{\varpi_l \zeta_v(r) r}{1 + \varpi_l \zeta_v(r) } \dd{r}}   
  \nonumber \\
&
\overset{(a)}{\geq} 1 - 2\pi\lambda_b  \int_{r_0}^{\infty}   
  \frac{ \frac{\vartheta d_0^{\alpha_l} G(r,h) d^{-\alpha_v} }{ G(r_0,h) }}{1 + \frac{\vartheta d_0^{\alpha_l} G(r,h)  d^{-\alpha_v}}{ G (r_0,h) } } r \dd{r}
   \nonumber \\
&\overset{}{=} 1 - 2\pi\lambda_b  \int_{r_0}^{\infty}   
  \frac{ \vartheta d_0^{\alpha_l} G(r,h) d^{-\alpha_v} }{G(r_0,h) + \vartheta d_0^{\alpha_l} G(r,h)  d^{-\alpha_v}  } r \dd{r}, 
 % \nonumber  
 \end{align}
 where (a) follows from $e^{-x}\geq 1- x$. Recall that $G(r,h)=\frac{1}{N_t}\frac{{\rm sin}^2\frac{N_t\pi}{2}({\rm sin}({\rm arctan}(\frac{h}{r})))}{{\rm sin}^2\frac{\pi}{2}({\rm sin}({\rm arctan}(\frac{h}{r})))}$.  
 From \cite{rajan2006efficient}, we have ${\rm arctan}(\frac{h}{r}) \simeq \frac{4h}{\pi r}$ and ${\rm sin}(\frac{4h}{\pi r}) \simeq \frac{4h}{\pi r}$ for $h\ll r$, since ${\rm sin}(x) \simeq x$ for small $x$. Also, ${\rm sin}^2\frac{N_t\pi}{2}(\frac{4h}{\pi r}) \simeq (\frac{2N_th}{ r})^2$ for small $N_t$ and $h\ll r$.
\begin{align} 
  \label{asypmtot}
 \mathbb{P}_{c|r_0} &\overset{}{\simeq} 
 1 - 2\pi\lambda_b  \int_{r_0}^{\infty}   
  \frac{ \vartheta (d_0^{\alpha_l}/d^{\alpha_v}) N_t }{N_t + \vartheta (d_0^{\alpha_l}/d^{\alpha_v}) N_t } r \dd{r}
  \nonumber \\  
  &\overset{}{=} 
  1 - 2\pi\lambda_b  \int_{r_0}^{\infty}   
  \frac{ \vartheta (d_0^{\alpha_l}/d^{\alpha_v})  }{1 + \vartheta (d_0^{\alpha_l}/d^{\alpha_v})  } r \dd{r}.   
 \end{align}
 Since $\mathbb{P}_{c|r_0}$ in (\ref{asypmtot}) is not a function of $N_t$, the UAV-UE coverage probability does not scale asymptotically with $N_t$, and the effect of $N_t$ on the coverage probability is minor. 
 % the coverage probability
\end{proof}

% has a marginal effect on
%&  \int_{r=r_0}^{\infty}   
%  \frac{ \frac{\vartheta 
%  {\rm sin}^2\frac{N_t\pi}{2}({\rm sin}({\rm arctan}(\frac{h}{r})))
%   {\rm sin}^2\frac{\pi}{2}({\rm sin}({\rm arctan}(\frac{h}{r_0})))
%   }{ 
%   {\rm sin}^2\frac{\pi}{2}({\rm sin}({\rm arctan}(\frac{h}{r})))
% {\rm sin}^2\frac{N_t\pi}{2}({\rm sin}({\rm arctan}(\frac{h}{r_0})))
%     } (\frac{d_0}{d})^{\alpha_n}}{1 + 
%    \frac{\vartheta 
%  {\rm sin}^2\frac{N_t\pi}{2}({\rm sin}({\rm arctan}(\frac{h}{r})))
%   {\rm sin}^2\frac{\pi}{2}({\rm sin}({\rm arctan}(\frac{h}{r_0})))
%    }{ 
%   {\rm sin}^2\frac{\pi}{2}({\rm sin}({\rm arctan}(\frac{h}{r})))
% {\rm sin}^2\frac{N_t\pi}{2}({\rm sin}({\rm arctan}(\frac{h}{r_0})))
%     } (\frac{d_0}{d})^{\alpha_n}
%    } r \dd{r}

%
% \mathbb{P}_{c} &=
% =2\pi\lambda_b   \int_{r_0=0}^{\infty}  r_0 e^{-2\pi\lambda_b \int_{r=r_0}^{\infty} 
%  \frac{\varpi_n \zeta_n(r) r}{1 + \varpi_n \zeta_n(r) } \dd{r}}   
%e^{-\pi\lambda_br_0^2} \dd{r_0}
% 
% Next, we derive the coverage probability of a static UAV-UE under \ac{HARP} association.
\vspace{-0.3 cm}
\subsection{Highest Average Received Power Association}
\label{sec:harp-assoc}
\vspace{-0.1 cm}
Each UAV-UE is associated to the BS that delivers the highest average power, i.e., the effect of small scale fading is averaged. Hence, the received signal power depends on the path loss $d^{-\alpha_v}$ and antenna gain $G(r,h)$, which is, in turn, determined by the elevation angle $\theta$. In Fig.~\ref{fadhil}, we plot the path loss, antenna gain, and the overall link gain, i.e., the path loss times antenna gain, versus the horizontal distance $r$. From (\ref{antenna-gain}), the antenna gain consists of $\frac{N_t}{2}$ lobes whose peaks increase with the horizontal distance $r$. Fig.~\ref{fadhil} also shows that the overall link gain (dashed line) consists of $\frac{N_t}{2}$ lobes whose maximum gains decrease as $r$ increases.

%Particularly this overall gain represents the large-scale fading $d^{-\alpha_v}$ times the antenna gain $G(r,h)$. Fig.~\ref{harp-assoc-r} shows that the antenna gain (red line) $G(r,h)$ consists of $\frac{N_t}{2}=\frac{8}{2}=4$ consecutive lobes, where the maximum \blue{seen} gain of each lobe monotonically increases with $r$ (i.e., decreases with $\theta$). Particularly, the maximum antenna gain is seen at $r=\infty$ ($\theta=0^\circ$) and the lowest is at $r=0$ ($\theta=\frac{\pi}{2}$, i.e., the UAV-UE us directly above the BS location). for \ac{HARP} association
% along with $r$	average power 
% structure 
%  Note that the idea behind the SCA method is
% We assume the peak of the lobe divides it into right and left sub-lobes. Assuming that there are BSs in an arbitrary lobe. If each lobe is seen as two sub-lobes, the BS that deliver the highest average power is considered the nearest BS far than the peak of the lobe, or the farthest BS before this peak. Since it is difficult to characterize all these cases, we consider the right half of the lobe with doubled density The idea behind this approximation is 

 It is worth highlighting that obtaining the serving distance \ac{PDF} is not possible given the non-convexity of the locations of BSs that deliver the highest power along $r$ (see Fig.~\ref{fadhil}). We hence propose a \emph{novel geometry-based approximation} that leverages the relative symmetrical characteristics of the overall link gain. The key idea behind this approximation is to reproduce an equivalent \ac{PPP} deployment of the BS locations in which the distance \ac{PDF} to the BS delivering the highest average power (within each lobe) can be obtained. With this in mind, we can then obtain an expression approximating the UAV-UE coverage probability under \ac{HARP} association. We particularly divide the space $r>0$ into $\frac{N_t}{2}$ regions, each of which corresponds to the boundary of one lobe. Then, for each lobe, we consider the zone from its peak to the next null, assuming doubled density of BSs. For instance, in Fig.~\ref{fadhil}, BSs of density $\lambda=2\lambda_b$ only exists within $r_{pj}$ to $r_{nj}$, $j\in\{1,\dots,\frac{N_t}{2}\}$. UAV-UEs then associate to the nearest BS from the reproduced deployment, \black{which delivers the highest average power within its lobe. However, since this BS might not be the one that delivers the highest average power among all the BSs, the adopted method is an approximation.} 
 % hence, the adopted method is an approximation.
 %
 % The underlying reason for applying this approximation is to exploit the symmetry of the link gain with $r$ to avoid calculating the exact serving distance \ac{PDF}.  It is unfortunate that \dots

From (\ref{antenna-gain}), there are $\frac{N_t}{2}$ nulls between the lobes that occur at elevation angles $\theta_{nj}$, where $\frac{N\pi}{2}{\rm sin}(\theta_{nj}) = j\pi$. Hence, $\theta_{nj}= {\rm arcsin}(\frac{2j}{N})$ and the equivalent horizontal distances at these nulls will be $r_{nj} = h/{\rm tan}\big({\rm arcsin}(\frac{2j}{N})\big)$. For instance, $\theta_{n0}=0^{\circ}$ and $r_{n4}=\frac{h}{{\rm tan}(0)} = \infty$ (where the indices are reversed for the distance). To obtain the locations of peaks $r_{pj}$, we define the overall link gain as $L(\theta)=G(\theta)\times d^{-\alpha_v}$:  
\begin{align}		% G(\theta) \big(r^2 + h^2\big)^{-\alpha_v/2}
&L(\theta) = \frac{1}{N_t} \frac{{\rm sin}^2\frac{N_t\pi}{2}\big({\rm sin}(\theta)\big)}{{\rm sin}^2\frac{\pi}{2}\big({\rm sin}(\theta)\big)} \times \big(r^2 + h^2\big)^{-\alpha_v/2}
\nonumber \\
&\overset{(a)}{=} \frac{1}{N_t h^{\alpha_v} } \frac{{\rm sin}^2\frac{N_t\pi}{2}\big({\rm sin}(\theta)\big)}{{\rm sin}^2\frac{\pi}{2}\big({\rm sin}(\theta)\big) \big({\rm tan}^{-2}(\theta) + 1\big)^{\alpha_v/2}},
\nonumber
\end{align}
where (a) follows from ${\rm tan}(\theta)=\frac{h}{r}$. We find elevation angles $\theta_{pj}$ at which the peaks occur by taking the first derivative of $L(\theta)$ and set it to zero. The first derivative is obtained as: 
\begin{align}
\frac{\partial L(\theta)}{\partial \theta} &= (N_t-1){\rm sin}(\frac{\pi(N_t+1)}{2}{\rm sin}(\theta)) +
\nonumber \\
&(N_t+1){\rm sin}(\frac{\pi(1-N_t)}{2}{\rm sin}(\theta)) + 
\nonumber \\
& \frac{\alpha_v\big({\rm cos}(\frac{\pi(N_t-1)}{2}{\rm sin}(\theta)) - {\rm cos}(\frac{\pi(N_t+1)}{2}{\rm sin}(\theta)\big)}{\pi{\rm cos}(\theta){\rm tan}(\theta)
\big(1+{\rm tan}^2(\theta)\big)}=0.
\nonumber 
 \end{align}

 \begin{figure} [!tb]	%[!t] %%    [htbp]
  \vspace{-0.3 cm}
\centering
\includegraphics[width=0.75 \linewidth]{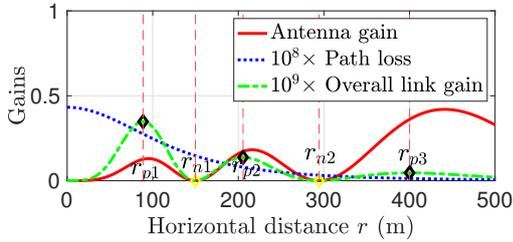}		%letter	roro2		
\caption {Illustration of the geometry-based approximation to calculate the UAV-UE coverage probability under \ac{HARP} association.}		% \magenta{($v=n$)}   % $N_t=8$, 
\label{fadhil}
 \vspace{-0.5 cm}
\end{figure}

The roots of this equation are $\theta_{pj}$, and equivalent distances  are $r_{pj} =\frac{h}{{\rm tan}(\theta_{pj})}$. The density of BSs within [$r_{jp},r_{jn}]$ is $\lambda=2\lambda_b$, and zero otherwise. From the \ac{PPP} definition, the probability that at least one BS exists within  $[r_{pj},r_{nj}]$ in the reproduced deployment is $1 - e^{-\pi\lambda(r_{nj}^2-r_{pj}^2)}$. Hence, the piece-wise serving distance \ac{CDF} will be $F_{R_m}(r_m)=$ 
\[
    \left\{
                \begin{array}{ll}
                  \Big(1 - e^{-\pi\lambda(r_{m}^2-r_{p1}^2)}\Big)\Big(1 - e^{-\pi\lambda(r_{n1}^2-r_{p1}^2)}\Big), 
                  \quad\quad j=1\\ 
                 \Big(1-e^{-\pi\lambda(r_{m}^2-r_{pj}^2)}\Big) \Big(1-e^{-\pi\lambda(r_{nj}^2-r_{pj}^2)}\Big)\times \\
                 \prod_{i=1}^{j-1} e^{-\pi\lambda(r_{ni}^2-r_{pi}^2)}, 
                 \quad\quad\quad\quad\quad\quad\quad\quad\quad
                 1<j\leq \frac{N_t}{2}
                 \\
                     0,\quad\quad\quad\quad\quad\quad
                     \quad\quad\quad\quad\quad\quad\quad\quad\quad\quad\quad{\rm otherwise}
                \end{array}   
              \right.
  \]
where the product term represents the probability that no BSs exist in the previous lobes. Hence, the serving distance \ac{PDF} will be $f_{R_m}(r_m)=\frac{\partial F_{R_m}(r_m)}{\partial r_m}=$ 
\[
    \left\{
                \begin{array}{ll}
              2\pi\lambda r_m e^{-\pi\lambda(r_{m}^2-r_{p1}^2)}\Big(1 - e^{-\pi\lambda(r_{n1}^2-r_{p1}^2)}\Big), 
                  \quad\quad j=1\\ 
             2\pi\lambda r_m e^{-\pi\lambda(r_{m}^2-r_{pj}^2)} \Big(1-e^{-\pi\lambda(r_{nj}^2-r_{pj}^2)}\Big)\times \\
                 \prod_{i=1}^{j-1} e^{-\pi\lambda(r_{ni}^2-r_{pi}^2)}, 
                 \quad\quad\quad\quad\quad\quad\quad\quad\quad
                 1<j\leq \frac{N_t}{2}
                 \\
                     0,\quad\quad\quad\quad\quad\quad
                     \quad\quad\quad\quad\quad\quad\quad\quad\quad\quad\quad{\rm otherwise.}
                \end{array}   
              \right.
  \]
It can be easily verified that $\int_{r_m=0}^{\infty}f_{R_m}(r_m)=1$. Given $f_{R_m}(r_m)$, the approximated coverage probability under \ac{HARP} association is  characterized in the next corollary. 			
% \red{From (\ref{cov-prob-near}) of Theorem \ref{lem-stat-near}}
% the serving distance \ac{PDF}
%\end{corollary}
%\magenta{The proof of Corollary \ref{geom-approx} proceeds in a similar way to Theorem \ref{lem-stat-near}, with the details omitted for brevity}. 
%
% $\varpi_l=\frac{\vartheta w_0^{\alpha_l}m_l}{A_lG(w_0,h)}$, $G=\frac{1}{N_t}\frac{{\rm sin}^2\frac{N_t\pi}{2}({\rm sin}({\rm arcsin}(\frac{h}{w_0})))}{{\rm sin}^2\frac{\pi}{2}({\rm sin}( {\rm arcsin}(\frac{h}{w_0}) ))}$, \blue{$h=h_d-h_{\rm BS}$}, and $\mathbb{P}_{c|w_0}^{n}$ is calculated as in (\ref{los-nlos-LT0}) with the Laplace transform of interference given by $\Lc_{I_l}(\varpi_l)  = e^{- \gamma(\varpi_l)}$, 
%where $\gamma(\varpi_l)= 2\pi\lambda_b \sum_{j=\floor{r_0\sqrt{a\eta}}}^{\infty} \Pb_{l}(\varsigma)  
%\int_{\frac{j}{\sqrt{a\eta}}}^{\frac{j+1}{\sqrt{a\eta}}}  		% {\rm max}(r_0,
%  \Big(1 - \big( \frac{m_l}{m_l + \varpi_l \zeta_l(r)} \big)^{m_l}  \Big) r \dd{r}$.   
%

\begin{corollary}
\label{geom-approx}
The coverage probability of static UAV-UEs under \ac{HARP} association is approximated as: 
\begin{align} 
\label{cov-prob-near}
\mathbb{P}_{c} = \int_{r_m=0}^{\infty}  \mathbb{P}_{c|r_m} f_{R_m}(r_m) \dd{r_m},  
\end{align} 
where $\mathbb{P}_{c|r_m}$ is calculated similar to $\mathbb{P}_{c|r_0}$ in Theorem \ref{lem-stat-near},  with $\varpi_l=\frac{\vartheta d_m^{\alpha_l}m_l}{A_lG(r_m,h)}$, $d_m=\sqrt{h^2+r_m^2}$, and  
$\gamma(\varpi_l)=2\pi\lambda_b \times \sum_{j=0}^{\infty} \Pb_{l}(\frac{j}{\sqrt{a\eta}})  
\int_{\frac{j}{\sqrt{a\eta}}}^{\frac{j+1}{\sqrt{a\eta}}}  		
  \Big(1 - \big( \frac{m_l}{m_l + \varpi_l \zeta_l(r)} \big)^{m_l}  \Big) r \dd{r}$. 
\end{corollary}
% % {\rm max}(r_0,
% \floor{r_0\sqrt{a\eta}}
\black{The proof of Corollary \ref{geom-approx} follows from Theorem \ref{lem-stat-near}.} 
As discussed previously, the cell association of UAV-UEs and their overall performance  is essentially driven by the availability of \ac{LoS} links and the encountered antenna gain. Thus far, we particularly showed that the coverage probability of static UAV-UEs heavily depends on the antenna pattern, but it does not scale with the number of antenna elements. Next, we study a scenario in which the UAV-UEs can be mobile.

 \vspace{-0.1 cm}
\section{Coverage Probability of mobile UAV-UEs}
\label{1D-RWP-Cov}
\vspace{-0.1 cm}		% \cite{1233531}		see Fig.~\ref{system-model-comp0}
% This model effectively describes the motion of UAV-UEs during the take off and landing.
We consider vertically-mobile UAV-UEs that make frequent up and down movements with a fixed velocity $\bar{\nu}$ in the finite vertical region $[h_1,h_2]$. We refer to it as \emph{vertical 1D \ac{RWP}} mobility model. Similar stochastic mobility models are adopted for UAV-BSs and UAV-UEs in the recent works \cite{amer2019mobility} and \cite{8671460,banagar2019performance,banagar2019fund,banagar20193gpp}. The proposed mobility model works as follows: Initially, at time instant $t_0$, the UAV-UE is at an arbitrary altitude $h_0$ selected uniformly from the interval  $[h_1,h_2]$. Then, at next time epoch $t_1$, this UAV-UE at $h_0$ selects a new random waypoint $h_1$ uniformly in $[h_1,h_2]$, and moves towards it. Once the UAV-UE reaches $h_1$, it repeats the same procedure to find the next destination altitude and so on. Eventually, the steady-state  altitude distribution converges to a non-uniform distribution $F_{Z}(z)$. Note that random waypoints refer to the altitude of a UAV-UE at each time epoch, which is uniformly-distributed in $[h_1,h_2]$, while vertical transitions are the differences in the UAV-UE altitude throughout its trajectory. While the random waypoints are independent and uniformly-distributed by definition, the  random lengths of vertical transitions are statistically dependent. This is because the endpoint of one movement epoch is the starting point of the next epoch. In \cite{amer2019mobility}, we showed that $f_{Z}(z)=\frac{h_1 z+h_2 z -h_1 h_2 - z^2}{\hbar^3/6}$, $\forall h_1<z<h_2$, where $\hbar=h_2-h_1$. It can be easily verified that the mean altitude is $\Eb[Z]=\frac{h_1+h_2}{2}$. 

The coverage mobility of UAV-UEs is studied under nearest and \ac{HARP} associations. Vertically-mobile UAV-UEs under nearest association maintain their connection to the nearest BSs. However, since UAV-UEs under \ac{HARP} association connect to the BS delivering the highest average power, the serving BS might change with the UAV-UE altitude. This is because the elevation angle and, correspondingly, the BS antenna gain change with the UAV-UE altitude.   

%For the \ac{HARP} scenario, however, they connect to the BS delivering the highest average power. By changing the UAV-UE altitude, the elevation angle also changes and, accordingly, the serving BS might also become different.
\vspace{-0.1 cm}
 \subsection{Nearest Association}
 \vspace{-0.1 cm}
 \label{sec:near-assoc-mob}
 For vertically mobile UAV-UEs, the distance to the nearest \ac{BS} is denoted as $w_0=\sqrt{r_0^2+(z-h_{\rm BS})^2}$, with $w_0\geq h_0$, and $h_0= h_1-h_{\rm BS}$. For simplicity, we use $f_{Z}(z)$ to obtain the steady state 3D distance \ac{PDF} $f_{W_0}(w_0)$, rather than averaging over two \acp{RV} $Z$ and $R_0$ in the coverage probability calculation. Since $R_0$ and $Z$ are two independent \acp{RV}, $f_{R_0,Z}(r_0,z)=f_{R_0}(r_0) f_{Z}(z)$. Given that, we transform the two \acp{RV} $R_0$ and $Z$ to $W_0$ and find $f_{W_0}(w_0)$, with the details omitted due to the limited space. The equivalent 3D distance \ac{PDF} is obtained as $f_{W_0}(w_0)=2\pi\lambda_bw_0\Omega(h_1,h_2) e^{-\pi\lambda_b w_0^2}$, where $\Omega(h_1,h_2) = \frac{\psi(h_2) (1-\varkappa)  -\psi(h_1) (\varkappa+1)  + 2 \kappa(h_1,h_2) -  2 \kappa(h_2,h_1)}{2 \pi  \lambda_b ^{3/2} \hbar^3/3}$, $\psi(x)=\text{erfi}\left(\sqrt{\pi } \sqrt{\lambda_b } (x-h_{\rm BS})\right)$, $\varkappa=2 \pi  \lambda_b  (h_1-h_{\rm BS}) (h_{\rm BS}-h_2)$, and $\kappa(x,y)= \sqrt{\lambda_b } (x - h_{\rm BS})e^{\pi\lambda_b(y-h_{\rm BS})^2}$.

 Since $w_0$ replaces both $z$ and $r_0$, we set $\theta$ to ${\rm arcsin}(\frac{h}{w_0})$ in (\ref{antenna-gain}) to obtain the antenna gain $G(w_0,h)$, where $h=h_d-h_{\rm BS}$ and $h_d=\Eb[Z]=\frac{h_1+h_2}{2}$ is the mean flying altitude. The effect of the vertical mobility, i.e., altitude variation, and horizontal distance randomness are now captured by the \ac{RV} $W_0$.

We characterize the coverage probability of mobile UAV-UEs under nearest association in the next corollary (whose proof follows Theorem \ref{lem-stat-near}). 
\begin{corollary}
\label{coro-mob-near}
The coverage probability of mobile UAV-UEs under nearest association is given by:  
\begin{align} 
\label{cov-prob-near-mobile}
\mathbb{P}_{c} = \int_{h_0}^{\infty} \mathbb{P}_{c|w_0} f_{W_0}(w_0) \dd{w_0},
\end{align} 
where $\mathbb{P}_{c|w_0}$ is calculated as in (\ref{los-nlos-LT0}), with
 $\varpi_l=\frac{\vartheta w_0^{\alpha_l}m_l}{A_lG(w_0,h)}$. The Laplace transform of \ac{LoS} interference is given by $\Lc_{I_l}(\varpi_l)  = e^{- \gamma(\varpi_l)}$, 
where $\gamma(\varpi_l)= 2\pi\lambda_b \sum_{j=j_0}^{\infty} \Pb_{l}(s)  
\int_{{\rm max}(w_0,s)}^{t}  		% {\rm max}(r_0,
  \Big(1 - \big( \frac{m_l}{m_l + \varpi_l  A_l G(w,h) w^{-\alpha_l} } \big)^{m_l}  \Big) w \dd{w}$, $j_0 = \floor{\sqrt{w_0^2-h^2}\sqrt{a\eta}}$, $s = \frac{j}{\sqrt{a\eta + h^2}}$, $t=\frac{j+1}{\sqrt{a\eta + h^2}}$,  and $\Pb_{l}(j)$ is calculated from (\ref{prob-los}).   
\end{corollary}
% with $o=\max(0,\floor{\frac{j\sqrt{a\eta}}{1000}-1})$
%   \blue{- $w=\sqrt{h^2+r^2}$ to shift the limit of integrations.}
 % the proof of 		due to limited space
%\blue{The proof of Corollary \ref{coro-mob-near} proceeds in a similar way to Theorem \ref{lem-stat-near}, with the details omitted for brevity. It can be }
%
%In particular, what the Nakagami fading parameter $m_l$, antenna down-tilting angle,  and the collaboration distance $R_c$ entail for the performance of mobile UAV-UEs is similar to the that of static UAV-UEs. Finally, a simple lower bound on the mobile UAV-UE coverage probability can be obtained similar to Corollary \ref{ch5:cov-prob-lb}, with the detailed omitted due to space limitation. 
% \blue{The effect of the antenna gain on the coverage probability of mobile UAV-UEs in (\ref{cov-prob-near-mobile}) can be interpreted in a similar way to the static scenario in (\ref{lem-stat-near}). Moreover, conditioning on $W_0=w_0$, the yielded expression holds the same insights as for static UAV-UEs in Section \ref{sec:near-assoc}, i.e., the coverage probability  does not scale with $N_t$.}
 
% in (\ref{cov-prob-near-mobile})
The antenna gain effect on the coverage probability of mobile UAV-UEs can be interpreted in a similar way to the static scenario in (\ref{lem-stat-near}). Particularly, conditioning on $W_0=w_0$, the yielded expression holds the same insights as for static users, i.e., the coverage probability  does not scale with $N_t$.

\vspace{-0.1 cm}
\subsection{Highest Average Received Power Association}
\vspace{-0.05 cm}
\label{hand-off}	% Unlike the nearest association case, 
Mobile UAV-UEs under \ac{HARP} association are prone to frequent \emph{altitude handovers}. This happens when their trajectory crosses multiple peaks and nulls of the antennas' side-lobes. For instance, Fig.~\ref{system-model-comp0} shows that the UAV-UE at altitude $h_1$ associates to the right BS, while it turns to attach to the left BS at altitude $h_2$. This \emph{altitude handover} negatively impacts the performance of UAV-UEs as it results in dropped connections. In fact, the elevation angle $\theta=\arctan(\frac{z-h_{\rm BS}}{r_m})$ plays a crucial role on determining the serving BS. Hence, it essential to average over the random distance $R_m$ (as in Corollary \ref{geom-approx}) for every possible $Z$, to correctly account for the handover and cell selection. Motivated by this fact, we describe the coverage probability as stated in the next corollary.
%  probability as next corollary 
\begin{corollary}
The approximated coverage probability of mobile UAV-UEs under \ac{HARP} association is described as:  
\begin{align} 
\label{cov-prob-near-mobile}
\mathbb{P}_{c} = \int_{z=h_1}^{h_2}\int_{r_m=0}^{\infty} \mathbb{P}_{c|r_m,z} f_{R_m}(r_m) f_{Z}(z)\dd{r_m} \dd{z},
% \nonumber
\end{align} 
where $\mathbb{P}_{c|r_m,z}$ is calculated as $\mathbb{P}_{c|r_0}$ in Theorem \ref{lem-stat-near}, with $\varpi_l=\frac{\vartheta (r_m^2+(z-h_{\rm BS}))^{\alpha_l/2}m_l}{A_lG(r_m,z-h_{\rm BS})}$. The Laplace transform of  \ac{LoS} interference is $\Lc_{I_l}(\varpi_l)  = e^{- \gamma(\varpi_l)}$, where $\gamma(\varpi_l)= 2\pi\lambda_b \sum_{j=0}^{\infty} \Pb_{l}(j)  
\int_{\frac{j}{\sqrt{a\eta}}}^{\frac{j+1}{\sqrt{a\eta}}}  		% {\rm max}(r_0,
  \Big(1 - \big( \frac{m_l}{m_l + \varpi_l \zeta_l(r)} \big)^{m_l}  \Big) r \dd{r}$.  % with $o=\floor{j-1}$  \frac{j}{1000} 
\end{corollary}
% and $\zeta_l(r)= A_l G(r,z-h_{\rm BS})) \big(r^2 + (z-h_{\rm BS})^2\big)^{-\alpha_v/2}$

To account for the UAV-UE mobility in the coverage probability expression, similar to \cite{amer2019mobility} and \cite{8692749}, we consider a linear function that reflects the cost of handovers. Particularly, we define the mobile UAV-UE coverage probability as:
\begin{align}
 \label{mob-cov0}
\mathbb{P}_{c}(\bar{\nu},\lambda_b,\beta)  = % \Pb(\Upsilon \geq \vartheta, \bar{H})+(1-\beta) \Pb(\Upsilon \geq \vartheta,H), 
(1-\beta) \mathbb{P}_{c} +\beta\Pb(\bar{H}) \mathbb{P}_{c},  
\end{align}
where $\beta\in[0, 1]$ represents the handover cost and $\Pb(\bar{H})$ is the probability of no handover. The   handover probability is defined as the probability that an UAV-UE associates to a new BS rather than the serving BS  after a unit time. It is clear from (\ref{mob-cov0}) that, if $\beta=1$, the first term vanishes and the UAV-UE will be in coverage only if there is no handover associated with its mobility.

\vspace{-0.2 cm}
\section{Numerical Results}
\vspace{-0.1 cm}
\label{num-result}
For our simulations, we consider a network having the parameter values indicated in Table \ref{ch4:table:sim-parameter}. % Monte Carlo simulations are used to corroborate the developed mathematical model.

\begin{table}[!tb]
\vspace{-0.9 cm}
\caption{Simulation Parameters} % title of Table
\centering % used for centering table
\begin{tabular}{c c c} % centered columns (3 columns)
\hline\hline %inserts double horizontal lines
Description & Parameter & Value  \\ [0.5ex] % inserts table
%heading
\hline % inserts single horizontal line
LoS and NLoS path-loss exponents& $\alpha_{l}$, $\alpha_{n}$& 2.09, 3.75 \\ %[1ex] % [1ex] adds vertical space
LoS and NLoS path-loss constants& $A_{l}$, $A_{n}$& -41.1, \SI{-32.9}{dB} \\ %[1ex] % adds vertical space
% Antenna main lobe gain&$G_m$&\SI{10}{dB}\\		%10
Number of antenna elements&$N_t$&8\\		%10
\ac{LoS} and \ac{NLoS} fading parameters&$m_l$, $m_n$&3, 1\\
BS height and UAV-UE altitude&$h_{\textrm{BS}}$, $h_d$&\SI{30}{m}, \SI{150}{m}\\
UAV-UE lowest and highest altitudes&$h_1$, $h_2$&\SI{140}{m}, \SI{160}{m}\\
Mobile UAV-UE speed& $\bar{\nu}$& \SI{20}{kmh} \\
Environment blocking parameters&$a$, $\eta$, $c$& 0.6, \SI{500}{km}$^{-2}$, \SI{30}{m}\\		% \SI{200}{km}$^{-2}$
%Area fraction occupied by buildings &$a$& 0.3\\
%Mean number of buildings &$e$&\SI{200}{km}$^{-2}$\\
%Buildings height Rayleigh parameter &$c$&15\\
Density of \acp{BS} and $\sir$ threshold&$\lambda_{b}$, $\vartheta$&\SI{50}{km}$^{-2}$, \SI{-15}{\deci\bel} 
\\
% $\sir$ threshold&$\vartheta$&\SI{0}{\deci\bel}\\
% Down-tilt angle&$\varphi$& $0^\circ$, $30^\circ$\\
\hline %inserts single line		and \red{vertical beamwidth}	, $\Omega$
\end{tabular}
\label{ch4:table:sim-parameter} % is used to refer this table in the text
\vspace{-0.3cm}
\end{table}

 \begin{figure} [!tb]	%[!t] %%    [htbp]
  \vspace{-0.9 cm}
\centering
\includegraphics[width=0.3 \textwidth]{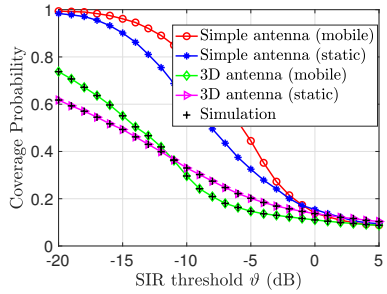}      %letter 	cov_prob_3D_near.eps
\caption {Coverage probability of static and mobile UAV-UEs under nearest association scheme versus the \ac{SIR} threshold $\vartheta$.} 
\label{cov-prob-near-assoc}
 \vspace{-0.4 cm}
\end{figure}

\vspace{-0.2 cm}
\subsection{Nearest Association}
 \vspace{-0.0 cm}
 \begin{figure} [!tb]	%[!t] %%    [htbp]
 \vspace{0.0 cm}
\centering
\includegraphics[width=0.3 \textwidth]{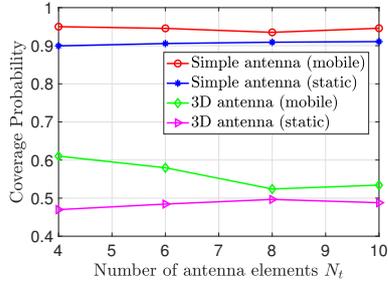}	%letter 	cov_prob_near_Nt.eps	
\caption {Coverage probability of static and mobile UAV-UEs under nearest association versus the number of antenna elements $N_t$.} %$N_t$
\label{static_cov_prob}
 \vspace{-0.1 cm}
\end{figure}
%%%%%%%%%%%%%%%%%%%%%%%%%%%%%%%%%%%%%%%%%

 %
%
We first evaluate the performance of static and mobile UAV-UEs under the nearest association scheme. We compare the UAV-UE coverage probability under practical antenna patterns with that adopting  simple antenna models, e.g., \cite{8692749,8713514,amer2019mobility}. Particularly, high-altitude UAV-UEs are essentially served and interfered from the antennas' side-lobes. Hence, for a simple antenna model, the antenna gain effect is normalized. Fig.~\ref{cov-prob-near-assoc} first shows that the  performance attained under a simple antenna model is superior to that from a practical antenna pattern.  This implies that adopting practical antenna models is vital to convey a realistic performance evaluation of the UAV-UEs. Fig.~\ref{cov-prob-near-assoc} also compares the performance of mobile UAV-UEs to static counterparts under nearest association. The effect of vertical mobility on the UAV-UE coverage probability is relatively marginal. This is attributed to the fact that there is no altitude handover associated with the vertical mobility as the UAV-UE maintains its connection with the nearest BS.

Fig.~\ref{static_cov_prob} investigates the prominent effect of the number antenna elements $N_t$ on the UAV-UE coverage probability. Fig.~\ref{static_cov_prob} shows that an increase in the number of antenna elements has an intangible effect on the coverage probability of static and mobile UAV-UEs under nearest association, hence verifying the claim in Corollary \ref{eff-of-Nt}. This also can be interpreted by the fact that, while increasing $N_t$ yields a higher number of lobes, the integrands in the coverage probability expression, which are functions of the antenna gains, constitute an overall area that does not significantly change with $N_t$.

\begin{figure} [!t]	%[!t] %%    [htbp]
 \vspace{-0.9 cm}
\centering
\includegraphics[width=0.30 \textwidth]{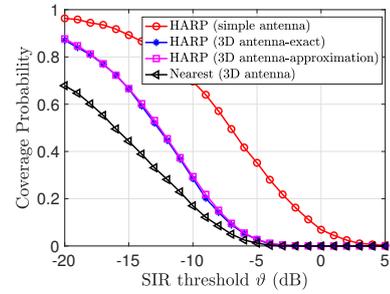}		%letter			
\caption {Coverage probability of static UAV-UEs under \ac{HARP} association ($m_v=1$, $N_t=4$).} 
\label{static_cov_prob_2}
 \vspace{-0.45 cm}
\end{figure}

 \vspace{-0.2 cm}
\subsection{HARP Association}
 \vspace{-0.1 cm}
Next, we discuss the performance of UAV-UEs under \ac{HARP} association. Fig.~\ref{static_cov_prob_2} verifies the accuracy of the geometry-based approximation of Corollary \ref{geom-approx}. Fig.~\ref{static_cov_prob_2} presents the exact and the obtained approximation of the static UAV-UE coverage probability versus the \ac{SIR} threshold $\vartheta$. Clearly, the adopted approximation is relatively tight. Fig.~\ref{static_cov_prob_2} also shows that the coverage probability of static UAV-UEs under practical antenna patterns is much reduced compared to the achievable coverage probability  from a simple antenna model.  

% tight, especially at high density of BSs and \blue{low number of antennas}
%\blue{Fig.~\ref{static_cov_prob_2}} illustrates the effect of ground BS density $\lambda_b$ on the UAV-UE coverage probability. While the coverage probability under the simple antenna assumption is independent of $\lambda_b$, as well established for \ac{PPP} modeling networks, it is slightly declines as $\lambda_b$ increases. This shows that \dots

%It is also noted that as long as the UAV-UE has frequent vertical movements, i.e., larger $\hbar$, the handover probability is lower since the effective horizontal travelled distance becomes shorter. The handover probability also monotonically increases with $\lambda_b$ since a higher rate of handover occurs for denser networks. Fig.~\ref{prob-of-ho-comp} shows the inter-CoMP handover probability versus the inter-cluster center distance $2R_h$. The handover probability monotonically decreases with $R_h$ since a lower rate of handover is anticipated when the cluster size increases. 
%% Having characterized the \ac{3D} RWP mobility model
%Next, we will use our proposed RWP model to obtain the coverage probability of \ac{3D} mobile UAV-UEs under the nearest association and CoMP transmission schemes.
% at different 

In Fig.~\ref{mob_cov_prob}, we plot the coverage probability of mobile UAV-UEs  versus the number of antenna elements $N_t$, at different penalty costs $\beta$. Fig.~\ref{mob_cov_prob} shows that the UAV-UE coverage probability monotonically decreases as both $N_t$ and $\beta$ increase. This can be interpreted by the fact that as long as $N_t$ increases, the mobile UAV-UE becomes more susceptible to handovers, which are penalized by the cost $\beta$. % In fact, as $N_t$ increases, higher handover rate is encountered, which lead to higher handover failures, which is represented by the penalty parameter $\beta$. % This also explains why the rate of decrease in the UAV-UE coverage probability grows linearly with $\beta$. 

Finally, Fig.~\ref{ho_vs_N} investigates the relation between the handover rate $H$ (\SI{}{sec^{-1}}) and the number of antenna elements $N_t$. The handover rate is numerically calculated from 
\begin{align}
H = \Eb\Big[ \frac{\#\text{ handoffs per each movement}}{\text{movement length}}\Big] \times \bar{\nu}, 
\nonumber  
\end{align}
where the movement is generated according to the adopted vertical mobility model. Fig.~\ref{ho_vs_N} shows that the handover rate monotonically grows with the  number of antenna elements. This is due to the higher number of nulls and peaks of the antenna vertical gain the UAV-UE crosses along its trajectory. \emph{From this result, we conclude  that a larger number of antenna elements yields a higher rate of altitude handovers.}
% This sequence of nulls and peaks directly leads to frequent handover execution.
% observation verifies the claim in Section \ref{hand-off}

 \begin{figure} [!t]	%[!t] %%    [htbp]	cov_prob_mob_uavue12.eps
  \vspace{0.0 cm}
\centering
\includegraphics[width=0.30 \textwidth]{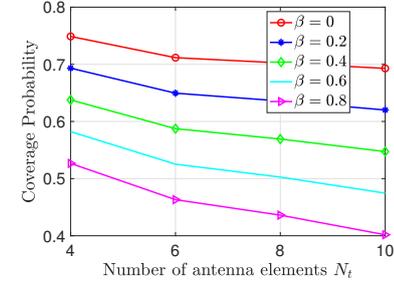}		%letter			
\caption {Coverage probability of mobile UAV-UEs under \ac{HARP} association ($h_d=\SI{100}{m}$,  $h_1=\SI{80}{m}$, $h_2=\SI{120}{m}$).}
\label{mob_cov_prob}
 \vspace{-0.4 cm}
\end{figure}

 \begin{figure} [!t]	%[!t] %%    [htbp]		ho_rate_vs_N2.eps
 \vspace{-0.9 cm}
\centering
\includegraphics[width=0.3 \textwidth]{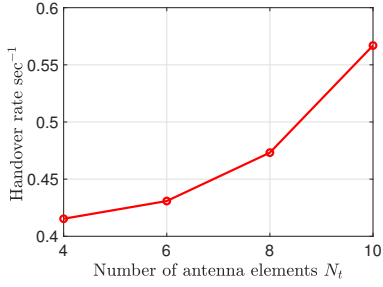} %letter. ho_rate_vs_N2.eps			
\caption {The handover rate versus the number of antenna elements ($h_d=\SI{100}{m}$, $h_1=\SI{80}{m}$, $h_2=\SI{120}{m}$).}  % \magenta{$\lambda_b$}
\label{ho_vs_N}
\vspace{-0.4 cm}
\end{figure}

% we observe that the integrand decreases with ?, hence verifying the claim.
%As illustrated in Fig.~\ref{ho_vs_N}, for $h_d<h_{{\rm BS}}$, the performance of the UAV-UE is maximized at certain $\varphi$, and beyond that it starts to degrade. However, for ground \acp{UE}, their performance is maximized at a higher $\varphi$. \magenta{Hence, adjusting the antennas' down-tilt angle yields a tradeoff between the performance of AUs and GUs owing to the difference in their altitudes}. 

% \begin{figure} [!tb]	%[!t] %%    [htbp]
%\centering
%\includegraphics[width=0.40\textwidth]{Figures/ch5c/ho_rate_vs_speed.eps}		%letter			
%\caption {\red{Handover rate of mobile UAV-UEs versus vertical speed.}}
%\label{ho_vs_speed}
% \vspace{-0.4 cm}
%\end{figure}
%
%\blue{In Fig.~\ref{ho_vs_speed}, the handover rate is plotted versus the vertical speed of mobile UAV-UEs.} This finding can be explained as follows: when $\varphi$ is small, an UAV-UE at an altitude $h_d>h_{{\rm BS}}$ can be served from the main lobe of tagged BS while also experiencing high interference from the main lobe of other interfering BSs.  Gradually, as $\varphi$ increases, the worst performance is observed when the UAV-UE is no longer served from the main lobe of tagged BS antennas while still experiencing high interference from the main lobe of other BSs. Finally, for very large $\varphi$, both intended and interfering signals stem from the side-lobes, and hence the performance is improved again. %of ground BS, which 

 \vspace{-0.2 cm}
\section{Conclusion} 
In this paper, we have studied the performance of UAV-UEs under practical antenna configurations. The coverage probability of static and mobile UAV-UEs are characterized as functions of the system parameters, namely, the number of antenna elements, density of BSs, and the UAV-UE altitude for both nearest and \ac{HARP} associations. The handover rate is also investigated to reveal the impact of practical antenna pattern on the cell association of mobile UAV-UEs. We have shown that the overall performance under practical antenna patterns is worse than that attained from a simple antenna model. Moreover, for static UAV-UEs, or mobile static UAV-UEs undergoing nearest association, the increase of the number of the antenna elements is shown to have a slight impact on their  achievable performance.  Conversely, for \ac{HARP} association, the coverage probability of mobile UAV-UE decreases as the number of antenna elements increases  due to the excessive rate of \emph{altitude handover} that is penalized by the handover cost $\beta$. % \red{CoMP and caching} 

\vspace{-0.2 cm}
\begin{appendix}
\vspace{-0.1 cm}
We first let $d=\sqrt{h^2+r^2}$ denote the communication link distance. The conditional coverage probability is given by:
\begin{align}
&\mathbb{P}_{c|r_0} = \Pb\Big(\frac{ \Upsilon_{|r_0}  }{I}>\vartheta\Big) = 
\Pb\Big(\frac{\chi_0 \zeta_l(r_0)  }{I}>\vartheta\Big)  
\nonumber \\
&
 % = \Pb\Big(\frac{\chi_0 A_l G(r_0,h) d_0^{-\alpha_l} }{I}>\vartheta\Big) 
= \Pb\Big(\chi_0>\frac{\vartheta d_0^{\alpha_l}}{A_l G(r_0,h)} I\Big) 
%\nonumber \\
  %\label{prob-y2} %
 %&
 \overset{(a)}{=}   \sum_{i=0}^{m_l-1}  \Big(\frac{\vartheta d_0^{\alpha_l}m_l}{A_l G(r_0,h)}\Big)^i \frac{1}{i!}I^i 
 \times
  \nonumber 
  \\ & 
%  \end{align}
%  \begin{align}
 e^{- \frac{\vartheta d_0^{\alpha_l}m_l}{A_l G(r_0,h)} I} \overset{(b)}{=}  \Big[ \sum_{i=0}^{m_l-1}  \frac{(-\varpi_l)^i}{i!}
\frac{\partial^i}{\partial\varpi_l^i} \Lc_{I}(\varpi_l)  \Big],
\end{align}
where (a) follows from the \ac{PDF} of the Gamma \ac{RV} $\chi_0$, (b) follows from the Laplace transform of aggregate interference, i.e., the \ac{RV} $I$. The \ac{LoS} and \ac{NLoS} interfering BSs are distributed independently of one another \cite{8713514}. Let $\lambda_v'(r)=2\pi\lambda_b r\Pb_{v}(r)$, $v\in\{l,n\}$, be the density of the inhomogeneous \acp{PPP} modeling each kind of the two interfering BS types. Hence, the Laplace transform of interference can be separated into a product of the Laplace transforms of the \ac{LoS} and \ac{NLoS} interference. Accordingly, the coverage probability will be \cite{8713514}:  
\begin{align}
\mathbb{P}_{c|r_0} & \overset{}{=}   \sum_{i=0}^{m_l-1} \frac{(-\varpi_l)^i}{i!}  \sum_{i_n+i_l=i}
\frac{i!}{i_n!i_l!}
 \frac{\varpi_l^{i_n}}{i_n!}
\frac{\partial^{i_n}}{\partial\varpi_l^{i_n}} \Lc_{I_n}(\varpi_l)  
\nonumber \\
\label{los-nlos-LT10}
& \quad \quad\quad \quad\quad \quad\quad \quad  \quad\quad \quad
\times \frac{\varpi_l^{i_l}}{i_l!}
\frac{\partial^{i_l}}{\partial\varpi_l^{i_l}} \Lc_{I_l}(\varpi_l)  ,
\end{align}
where $\Lc_{I_l}$ and $\Lc_{I_n}$ are the Laplace transforms of the \ac{LoS} and \ac{NLoS} interference, respectively.  The second sum is over all of the combinations of non-negative integers $i_l$ and $i_n$. We obtain the Laplace transform of \ac{LoS} interference $\Lc_{I_l}(\varpi_l)$, whereas the \ac{NLoS} interference follows in the same manner.   %, $v\in\{l,n\}$  
% The Laplace transform of \ac{LoS} interference is given by: 
\begin{align}
& \Lc_{I_l}(\varpi_l)  = \Eb_{I_l} 
\Big[ e^{-\varpi_l I_l}\Big]
= \mathbb{E} \Bigg[
 e^{-\varpi_l\sum_{r \in \Phi_b^{!}}  \chi_r  \zeta_l(r)}  
  \Bigg] 
\nonumber \\
 &
 = \mathbb{E}_{\Phi_b,\chi_r} 			% \mathbb{E}_{}
  \Bigg[  
  \prod_{r \in \Phi_b^{!}}  e^{- \varpi_l \chi_r  \zeta_l(r)}
  \Bigg] 
%  \nonumber \\ 
%   &
   = \mathbb{E}_{\Phi_b} 			% \mathbb{E}_{}
  \Bigg[  
  \prod_{r \in \Phi_b^{!}} \mathbb{E}_{\chi_r} e^{- \varpi_l \chi_r  \zeta_l(r)}
  \Bigg] 
  \nonumber \\
  & 
   \overset{(a)}{=} 
\mathbb{E}_{\Phi_b} 			
  \Bigg[  
  \prod_{r \in \Phi_b^{!}}  \Big(\frac{m_v}{m_v+\varpi_l \zeta_l(r} \Big)^{m_v}
  \Bigg],  
%  = \mathbb{E}_{\Phi_b} 			% \mathbb{E}_{}
%  \Bigg[  
%  \prod_{r \in \Phi_b^{!}}  \frac{1}{\Big(1 + \frac{\varpi_l \zeta_l(r)   }{m_v} \Big)^{m_v}}
%  \Bigg] 
%  \nonumber 
%  \label{before-chi}
%        &\overset{(a)}{=} 
% {\rm exp}\Bigg(-2\pi \lambda_b\int_{\nu=R_c}^{\infty}\Big(1 - \mathbb{E}_{\chi} 
%       e^{-\varpi_l \chi  P(\nu)^2}
%        \Big)\nu\dd{\nu}\Bigg) 
%\\ 
%\label{LT_c1}      
%  &\overset{(b)}{=} {\rm exp}\Bigg(-2\pi \lambda_b\int_{\nu=R_c}^{\infty}\Big(1 -
%        \delta_l\Pb_{l}(\nu) - \delta_n\Pb_{n}(\nu) 
%        \Big)\nu\dd{\nu}\Bigg)
% \overset{(c)}{=}e^{\Omega(\varpi_l)_{|\boldsymbol{r}_{\kappa}}} ,
\end{align}
where $\Phi_b^{!}= \Phi_b\setminus r_0$ and (a) follows from $\chi_r\sim\Gamma(m_v,\frac{1}{m_v})$. Using the \ac{PGFL} of \ac{PPP} and the density of \ac{LoS} interferers $\lambda_l'(r)=2\pi\lambda_b r\Pb_{l}(r)$, the Laplace transform of \ac{LoS} interference will be $\Lc_{I_l}(\varpi_l)=$
\begin{align}
\label{LT-LoS10}
& = e^{-\overbrace{2\pi\lambda_b \int_{r_0}^{\infty} \Pb_{l}(r) 
  \Big(1 - \big( \frac{m_l}{m_l + \varpi_l \zeta_l(r)} \big)^{m_l} \Big) r \dd{r}}^{\gamma(\varpi_l)}}. 
%   \nonumber  
   \end{align}
From (\ref{prob-los}), $\Pb_{l}(r)$ is a step function of the serving distance $r$. We use this fact to separate  the integral above into a sum of weighted integrals, resulting in $\gamma(\varpi_l)=2\pi\lambda_b\times$
\begin{align}
&\sum_{j=\floor{r_0\sqrt{a\eta}}}^{\infty} \Pb_{l}(\varsigma)  
\int_{{\rm max}(r_0,\frac{j}{\sqrt{a\eta}})}^{\frac{j+1}{\sqrt{a\eta}}}  			% {\rm max}(
  \Big(1 - \big( \frac{m_l}{m_l + \varpi_l \zeta_l(r)} \big)^{m_l} \Big) r \dd{r}. 
    \nonumber 
%    \\
%&=2\pi\lambda_b \hspace{-4mm}\sum_{j=\floor{r_0\sqrt{a\eta}}}^{\infty} \hspace{-4mm}\Pb_{l}(\varsigma) \hspace{-2mm} 
%\int_{{\rm max}(r_0,\frac{j}{\sqrt{a\eta}})}^{\frac{j+1}{\sqrt{a\eta}}}  			% {\rm max}(
% \hspace{-2mm}  \Big(1 - \big(1- \frac{1}{1 + m_l(\varpi_l \zeta_l(r))^{-1}} \big)^{m_l} \Big) r \dd{r},
%    \nonumber
\end{align}
Substituting (\ref{LT-LoS10}) into (\ref{los-nlos-LT10}) yields the conditional coverage probability $\mathbb{P}_{c|r_0}$. Given the nearest distance \ac{PDF} $f_{R_0}(r_0)$, the unconditional coverage probability is obtained.

\textcolor{white}{\cite{baza1,baza2,baza3,baza4,baza5,baza6,baza7,baza8,baza9,baza10,baza11,baza12,baza13}}
% Caching Papers		% amer2019optimizing			amer2020cooperative
\textcolor{white}{\cite{Delay-Analysis,8647532,amer2020joint,8886101,8412262,amer2018optimizing,amer2019performance,amer20200caching,chaccour2019Reliability}}	
% UAV Papers
%\textcolor{white}{\cite{amer2020caching,8756296}}
% VR paper
% \cite{chaccour2019Reliability}	
% Master papers
 \textcolor{white}{\cite{7440654,7605066,7925123}}

%Using the binomial series expansion, $\gamma(\varpi_l)$ can be simplified to 
%
%\begin{align}
%2\pi\lambda_b\hspace{-4mm} \sum_{j=\floor{r_0\sqrt{a\eta}}}^{\infty} \hspace{-4mm}\Pb_{l}(\varsigma)    			% {\rm max}(
% \sum_{k=1}^{m_l}{m_l \choose k}(-1)^{k+1}\hspace{-2mm}\int_{{\rm max}(r_0,\frac{j}{\sqrt{a\eta}})}^{\frac{j+1}{\sqrt{a\eta}}}\underbrace{ \frac{r}{(1 + m_l(\varpi_l \zeta_l(r))^{-1})^k}}_{\chi}\dd{r},
%    \nonumber
%\end{align}

% \red{Using the Jensen's inequality}. 
%

\end{appendix}

\vspace{-0.6 cm}
\bibliographystyle{IEEEtran}
 %\bibliography{bibliography}
 \bibliography{public_files/bibliography}
\end{document}